\documentclass[12pt]{amsart}
\usepackage{hyperref, amsmath,amssymb,amsfonts,amsbsy,mathtools,cite,setspace}
\usepackage{graphicx}
\usepackage{amssymb}
\usepackage{amsthm}
\usepackage{enumerate}
\usepackage{wasysym}
\usepackage{cite}
\usepackage{tikz}
\usepackage{mathtools}

\usepackage{braket}
\usepackage{color}

\def\ket#1{| #1 \rangle}
\def\bra#1{\langle #1 |}
\def\kb#1#2{|#1\rangle\!\langle #2 |}
\def\bk#1#2{\langle #1 |#2\rangle}

\def\be{\begin{eqnarray}}
\def\ee{\end{eqnarray}}
\def\bee{\begin{eqnarray*}}
\def\eee{\end{eqnarray*}}

\newtheorem{defn}{Definition}
\newtheorem{prop}{Proposition}
\newtheorem{thm}{Theorem}
\newtheorem{exa}{Example}
\newtheorem{lem}{Lemma}
\newtheorem{con}{Conjecture}
\newtheorem{rmk}{Remark}
\newtheorem{cor}{Corollary}

\newcommand{\C}{{\mathbb C}}

\newcommand{\M}{{\mathbb M}}
\renewcommand{\H}{{\mathcal H}}
\renewcommand{\S}{{\mathcal S}}
\newcommand{\operp}{$\bigcirc$\kern-.91em{$\perp$}}

\newcommand{\tr}{\operatorname{Tr}}

\def\be{\begin{eqnarray}}
\def\ee{\end{eqnarray}}
\def\bee{\begin{eqnarray*}}
\def\eee{\end{eqnarray*}}
\def\ot{\otimes}

\begin{document}

\title[Quantum Error Correction and One-Way LOCC]{Quantum Error Correction and One-Way LOCC State Distinguishability}
\author[D.W.Kribs, C.Mintah, M.Nathanson, R.Pereira]{David W. Kribs$^{1,2}$, Comfort Mintah$^{1}$, Michael Nathanson$^3$, Rajesh Pereira$^{1}$}

\address{$^1$Department of Mathematics \& Statistics, University of Guelph, Guelph, ON, Canada N1G 2W1}
\address{$^2$Institute for Quantum Computing and Department of Physics \& Astronomy, University of Waterloo, Waterloo, ON, Canada N2L 3G1}
\address{$^3$Department of Mathematics and Computer Science, Saint Mary's College of California, Moraga, CA, USA 94556}

\begin{abstract}
We explore the intersection of studies in quantum error correction and quantum local operations and classical communication (LOCC). We consider one-way LOCC measurement protocols as quantum channels and investigate their error correction properties, emphasizing an operator theory approach to the subject, and we obtain new applications to one-way LOCC state distinguishability as well as new derivations of some established results. We also derive conditions on when states that arise through the stabilizer formalism for quantum error correction are distinguishable under one-way LOCC.
\end{abstract}



\date{July, 2018}

\maketitle

\section{Introduction}

Quantum error correction as a field of study began with efforts in quantum computing and communication during the early 1990's \cite{sho95a,ste96a,bdsw96,gottesmanstab,knilllaflamme}. Initially the motivation was to understand how quantum information could be encoded into physical systems in ways that would allow for the preservation or controlled recovery of quantum bits, and thus overcome unwanted noise or errors brought on by physical intrusions such as decoherence. Over the subsequent two and a half decades, the subject of quantum error correction has blossomed and expanded to the extent that it now touches on almost every area of quantum information science. Here we note the structure theory that has been developed over the last several years~ \cite{klp05,klpl06,bkk07,bkk07b}, culminating in the theory of ``operator algebra quantum error correction''.

Elsewhere in quantum information, there is longstanding interest in  the problem of identifying states from a composite quantum system using only  operations local to each subsystem. Examples such as quantum teleportation and data hiding \cite{Teleportation, terhal2001hiding,eggeling2002hiding} can be seen through this lens, and the problem as a whole can be seen as a means to understand the interplay between locality and entanglement. This is a subset of the more general problem of determining which information tasks can be accomplished using only LOCC operations, indicating the restriction to Local Operations but allowing unrestricted Classical Communication between the subsystems. There is extensive literature on the problem of local operations, including the problem of state discrimination \cite{bennett1999quantum,ghosh2004distinguishability,horodecki2003local,chefles2004condition,cosentino2013small}. The set of permissible LOCC operations is notoriously messy, leading us to examine the more restrictive class of one-way LOCC operations, in which operations must be performed sequentially on the subsystems in a prescribed order. This restricted problem has been the subject of investigations along a number of different lines, and many fundamental LOCC results can  be achieved using only one-way LOCC \cite{Walgate-2000,Nathanson-2005,fan2004distinguishability,N13}.

Although there are clearly connections that have been made between studies in quantum error correction and LOCC, with initial explorations going back to the beginnings of the subjects~\cite{bdsw96}, a comprehensive investigation has yet to be undertaken that involves the considerable advances in both subjects over the past fifteen years.  In this paper, we initiate the first such investigation. We bring together some key aspects of the two subjects, with an emphasis on more recently developed approaches and results, and look for common themes and new results. As the theory of quantum error correction has been more extensively developed over time, our focus here is to look for applications of that theory to LOCC, and in particular to the fundamental problem of quantum state distinguishability in one-way LOCC. As in our recent work \cite{kribs2016operator}, our analysis focuses on underlying mathematical structures such as special types of operator systems and operator algebras that can now be seen as relevant to both fields.

This paper is organized as follows. In Section~2 we cast one-way LOCC measurement protocols as quantum channels and then start an analysis of these channels from a standard (subspace) error correction perspective, obtaining LOCC applications and a new view on quantum teleportation. In Section~3 we expand this investigation to the more general setting of correctable operator algebras and obtain further applications to one-way LOCC state distinguishability, specifically on the distinguishability of certain physically relevant projections. Independent of these results, in Section~4 we build on the operator approach of \cite{kribs2016operator}  to find sets of quantum states that are generated via the fundamental stabilizer formalism for quantum error correction and determine precisely when they are distinguishable under one-way LOCC.

We use standard quantum information notation and nomenclature throughout the paper, and introduce basic notions as they arise in the presentation.

\section{One-Way LOCC as a Quantum Channel}

In this section, we show how one-way LOCC protocols may be viewed in the quantum channel formalism. We then investigate how quantum error correction conditions relate to the channels and obtain applications back to one-way LOCC state distinguishability.

Suppose that two parties Alice and Bob share an unknown (entangled) state  from a known set of states, $\ket{\phi} \in \mathcal S = \{ \ket{\phi_i} \}$, on their combined quantum system $\mathcal H = \mathcal H_A \otimes \mathcal H_B$, with $a = \dim \mathcal H_A \le \dim \mathcal H_B < \infty$. We identify the elements of $\mathcal S$ with operators in the standard way, writing  $\ket{\phi_i} = (I \otimes B_i)\ket{\Phi}$, where $I$ is the identity operator on $\mathcal H_A$, the $B_i$ are operators mapping  $\mathcal H_A$ to  $\mathcal H_B$, and $\ket{\Phi} = \frac{1}{\sqrt{a}}(\ket{00}+\ldots + \ket{a-1 \, a-1}$ is the canonical maximally-entangled state on ${\mathcal H}_A \otimes {\mathcal H}_A$.

For instance, a case of central interest has $\mathcal H_A$ and $\mathcal H_B$ the same dimension and with maximally entangled states $\ket{\phi_i} = (I\otimes U_i) \ket{\Phi}$ with $U_i$ unitary operators on $\mathcal H_A$.

The task is to, if possible, determine the identity of $\ket{\phi}$ using only one-way LOCC measurement operations. That is, Alice will perform a local measurement on her system and report the results to Bob, who will then perform a measurement on his system. In this setting, Alice and Bob can always gain maximal information by performing local complete measurements on their system, so we will assume this without loss of generality. (See, e.g., \cite{kribs2016operator,N13}.) The first step of an optimal one-way LOCC protocol can be viewed as a quantum-classical channel acting on Alice's system:
\bee
\Phi_{QC}(\rho) = \sum_{ j  = 1}^r \kb{j}{j} \tr (\sigma_j \rho),
\eee
where $\sum_{j=1}^r \sigma_j = I_a$. Since $\{ \sigma_j \}$ forms a complete measurement on $\mathcal H_A$, each operator is rank one and so we may write
\bee
\sigma_j  = m_j \kb{\varphi_j}{\varphi_j}.
\eee
Then
\bee
\Phi_{QC}(\rho) = \sum_{ j  = 1}^r  V_{j} \rho V_{j}^*,
\eee
where $V_{j} = \sqrt{m_j} \kb{j}{\varphi_j}$. The effect of Alice's measurement on the whole system is the application of the channel $({\Phi_{QC}})_A \otimes \mathrm{id}_B$:
\bee
\Phi_{QC} \otimes \mathrm{id}_B : {\mathcal L}(\H_A \otimes \H_B) \rightarrow {\mathcal L}(\H_C \otimes \H_B),
\eee
where $\H_C$ is an $r$-dimensional  classical ancilla system $\H_C$ to emphasize the change. The Kraus operators of $\Phi_{QC} \otimes \mathrm{id}_B$ are given by $\{V_j \otimes I_B \}$.

Alice and Bob's task can be seen as trying to identify $\ket{\phi}\in \mathcal S$ despite the noise introduced by this channel. This suggests using tools from quantum error correction  to find code subspaces on which this channel's noise can be overcome. Applying the Knill-Laflamme conditions \cite{knilllaflamme,nielsen} directly to the channel $\Phi_{QC} \otimes \mathrm{id}_B$, a correctable subspace $\mathcal{C}$ will be given by a projection $P_\mathcal{C}$ on $\mathcal H_A \otimes \mathcal H_B$ such that there are complex scalars  $\lambda_{i,j}$ with
\bee
\lambda_{i,j}  P_\mathcal{C}=  P_\mathcal{C} (V_i^*V_j \otimes I_B) P_\mathcal{C}.
\eee

This leads us to a first statement relating quantum error correcting code spaces and one-way LOCC with maximally-entangled states.

\begin{thm}\label{thm: MaxEnt LOCC implies correctability}
Let $\mathcal S = \{ \ket{\phi_i}\}$ be a set of maximally-entangled orthogonal states in $\H_A \otimes \H_B$ that can be perfectly distinguished with one-way LOCC, with $|S| = d$ and $\dim \H_A \le \dim \H_B$.

If we define $\Phi_{QC}$ to be the quantum-classical channel corresponding to Alice's optimal measurement, then there exists a $d$-dimensional correctable code $\mathcal C$ in  $\H_A \otimes \H_B$ for the noise model ${\mathcal E} = \Phi_{QC} \otimes \mathrm{id}_B$, with the code space $\mathcal C = \mathrm{span}\,\{ \ket{\phi_i} \}$ generated by the elements of $\mathcal S$.
\end{thm}

\begin{proof}
We first define the projection,
\bee P_{\mathcal C} := \sum_{i = 1}^d \kb{\phi_i}{\phi_i}. \eee
With $V_j = \sqrt{m_j} \kb{j}{\varphi_j}$ as above, we have
\[
V_i^*V_j = \sqrt{m_im_j} \ket{\varphi_i}\bk{i}{j}\bra{\varphi_j} = \delta_{i,j} \kb{\varphi_j}{\varphi_j}
\]
and so we need only check that for each $j$ there are scalars $\lambda_j$ such that, \bee
 P_\mathcal{C} (\kb{\varphi_j}{\varphi_j} \otimes I_B) P_\mathcal{C} = \lambda_j P_\mathcal{C}. \eee
Since the measurement perfectly distinguishes the elements of ${\mathcal S}$, for each $j$, the (unnormalized)  states  $\{(\kb{\varphi_j}{\varphi_j}_A \otimes I_B) \ket{\phi_i}\}_i$ are mutually orthogonal. Hence,
\bee
P_{\mathcal C} (\kb{\varphi_j}{\varphi_j}_A
\otimes I_B) P_{\mathcal C}  &=& \sum_{i,k= 1}^d \kb{\phi_i}{\phi_i} (\kb{\varphi_j}{\varphi_j}_A
\otimes I_B) \kb{\phi_k}{\phi_k} \\
& = &  \sum_{i= 1}^d \kb{\phi_i}{\phi_i} (\kb{\varphi_j}{\varphi_j}_A
\otimes I_B) \kb{\phi_i}{\phi_i} .
\eee
The fact that the states in ${\mathcal S}$ are maximally-entangled, and also with $a = \dim\H_A \le \dim\H_B$, implies that each $\ket{\phi_i} = (I\otimes U_i)\ket{\Phi}$ with $U_iU_i^* = I_A$; hence $\bra{\phi_i} (\kb{\varphi_j}{\varphi_j}_A
\otimes I_B) \ket{\phi_i} = \frac{1}{a}$ for all $i$. Thus, for all $j$,
\bee
P_{\mathcal C} (\kb{\varphi_j}{\varphi_j}_A
\otimes I_B) P_{\mathcal C}  = \frac{1}{a} P_{\mathcal C},
\eee
and it follows that $\mathcal C$ is correctable for $\mathcal E = \Phi_{QC} \otimes \mathrm{id}_B$.
\end{proof}

For illustrative purposes let us consider a very simple case of this result.

\begin{exa}
Let $\mathcal H_A = \C^2 = \mathcal H_B$ and let $\mathcal S = \{ \ket{\phi_1}, \ket{\phi_2}\}$ be the Bell basis states $\ket{\phi_1} = \frac{1}{\sqrt{2}}(\ket{00} + \ket{11})$, $\ket{\phi_2} = \frac{1}{\sqrt{2}}(\ket{01} + \ket{10})$. These are maximally entangled states, observe $\ket{\phi_1}=\ket{\Phi}$ and $\ket{\phi_1} = (I\otimes X)\ket{\Phi}$ where $X$ is the Pauli bit flip operator, and they are perfectly distinguished by the one-way LOCC measurement
$\M = \{ A_k \otimes B_{k,j} \}_{j,k =1}^2$ with $A_1 = \kb00$, $A_2 = \kb11$, and $B_{i,j} = A_{2-\delta_{ij}}$ (and where $\delta_{ij}$ is the Kronecker delta).

Meshing with the proof above, here we have $a=2=r$, 
$\ket{\varphi_1}= \ket{0}$, 
$\ket{\varphi_2}= \ket{1}$, and then $V_1 = \kb00$, $V_2 = \kb11$. The relevant two-qubit channel $\Phi_{QC}\otimes \mathrm{id}_2$ in this case implements the von Neumann measurement with Kraus operators $\{ {V}^\prime_j := V_j \otimes I_2 \}_{j=1,2}$. The code space here is the single qubit subspace $\mathcal C = \mathrm{span}\,\{ \ket{\phi_1},\ket{\phi_2}\}$, with $P_{\mathcal C} = \kb{\phi_1}{\phi_1} + \kb{\phi_2}{\phi_2}$, which is correctable for $\mathcal E = \Phi_{QC}\otimes \mathrm{id}_2$ with conditions: $P_{\mathcal C} (V_i^* V_j \otimes I_2) P_{\mathcal C} = \frac{\delta_{ij}}{2} P_{\mathcal C}$.

Observe the operators ${V}^\prime_j$ act as scalar multiples of unitaries with mutually orthogonal ranges when restricted to $\mathcal C$; in particular, one can easily verify ${V}^\prime_1 \mathcal C = \mathrm{span}\,\{\ket{00},\ket{01}\}$ and ${V}^\prime_2 \mathcal C = \mathrm{span}\,\{\ket{10},\ket{11}\}$.  The correction operation in this case can be seen from the standard error correction recovery construction to be given by the channel $\mathcal R(\rho) = \sum_{k=1}^2 R_k \rho R_k^*$, with $R_k = W_k^* P_k$ and $P_1=\kb00 \otimes I_2$, $P_2=\kb11 \otimes I_2$, $W_k = \sqrt{2} \overline{V}_k = \sqrt{2} P_k$.
\end{exa}

Let us extrapolate from the end of this example to comment more generally on how the recovery operation works in the error correction protocols associated with this result. Suppose that Alice and Bob can distinguish the basis $\mathcal S = \{\ket{\phi_i} = (I \otimes B_i)\ket{\Phi}\} $ as described in the theorem. As these states are maximally-entangled, the corresponding operators $B_i$ are partial isometries from $\H_A$ into $\H_B$.  Suppose that Alice performs the measurement $\{  \sqrt{m_j} \kb{j}{ \overline{\varphi_j}} \}$,where $\ket{ \overline{\varphi_j}} $ is the entrywise complex conjugate of $\ket{ {\varphi_j}}$ in the standard basis. If Alice  gets the outcome $j=x$, then the state of Bob's system is $B_i \ket{\varphi_x}$ with probability $\frac{1}{a}$, implying that the states $\{ B_i \ket{\varphi_x}\}_{i = 1}^d$ are mutually orthogonal. This allows us to define a correction operator on Bob's system, $R_x =  \sum_i \kb{\phi_i}{\varphi_x} B_i^*$, with $R_xR_x^* = P_{\mathcal C}$ for each $x$. We can then build the full correction channel \bee \mathcal R(\tau) = \sum_{x=1}^a  (\bra{x}  \otimes R_x)\tau (\ket{x} \otimes R_x^*). \eee
Let $\ket{\psi} = \sum_i \alpha_i \ket{\phi_i}$ be any element of $\mathcal C = \mathrm{span}\,\{\ket{\phi_i}\}$ and set $\rho = \kb{\psi}{\psi}$. Then we have
\bee
\mathcal R \circ (\Phi_{QC} \otimes \mathrm{id}_B)(\rho) &=& \mathcal R \Big( \sum_{j} (V_j \ot I_B)\rho(V_j^* \ot I_B) \Big)  \\
& = & \frac{1}{a} \, \mathcal R \Big( \sum_{i,j,k} m_j \alpha_i\overline{\alpha_k}  \kb{j}{j} \otimes B_i \kb{ \varphi_j}{ \varphi_j}B_k^* \Big) \\
& = & \frac{1}{a}  \sum_{i,k,x} m_x \alpha_i\overline{\alpha_k}   R_xB_i \kb{ \varphi_x}{ \varphi_x}B_k^* R_x^* \\
& = & \frac{1}{a} \sum_x m_x  \big( \sum_{i,k} \alpha_i\overline{\alpha_k}   \kb{\phi_i}{\phi_k} \big) \\
& = & \rho ,
\eee
where the second equality above follows from the teleportation identity $(\kb{j}{\overline\varphi_j}\otimes I_A)\ket{\Phi} = \frac{1}{\sqrt{a}} \ket{j} \ket{\varphi_j}$; and the fourth equality follows from the mutual orthogonality of  the $\{ B_i \ket{\varphi_x}\}$. This means that if Alice and Bob share an unknown state in the code space $\mathcal C$, they can use the LOCC protocol to teleport Alice's half of the system to Bob's without requiring any additional entanglement.

The following result can be seen as a converse to Theorem \ref{thm: MaxEnt LOCC implies correctability}:

\begin{thm}\label{thm: correctability implies LOCC}
Let $\Phi$ be a quantum channel on $\mathcal{L}(\H_A)$ with Kraus operators $\{ A_j \}$ such that each $A_j$ is rank 1 and $A_j^*A_i = 0$ when $i \ne j$.

If there exists a correctable code $\mathcal C$ in  $\H_A \otimes \H_B$ for the noise model ${\mathcal E} = \Phi \otimes \mathrm{id}_B$, then any basis of $\mathcal C$ can be distinguished with one-way LOCC starting with Alice applying the measurement $\{ A_j \}$.\end{thm}

\begin{proof}
Since each $A_j$ is rank one and $A_j^*A_i = 0$ for $i \ne j$, we can write $A_j = \sqrt{m_j} \kb{j}{ \overline{\varphi_j}}$. As $\mathcal C$ is correctable for $\mathcal E$, from the error correction conditions we have scalars $\lambda_j$ such that for every $j$,
\bee
P_{\mathcal C}(A_j^* \ot I_B)(A_j \ot I_B)P_{\mathcal C} = \lambda_j P_{\mathcal C}.
\eee
If ${\mathcal S} = \{ \ket{ \phi_i} = (I \ot M_i)\ket{\Phi} \}$ is a basis for ${\mathcal C}$, then this implies that for  every $j$ and for $i \ne k$,
\bee
0= \bra{\phi_i} (A_j^*A_j\ot I_B) \ket{\phi_k} = \frac{1}{d} \bra{{\varphi_j}} M_i^*M_k \ket{{\varphi_j} } .
\eee
If we suppose that Alice and Bob's system starts in one of the states in ${\mathcal S}$, then
when Alice measures and gets the result $j$,  Bob's system will be in the unnormalized state $M_i\ket{\varphi_j} $ for some $i$. Since these states are mutually orthogonal, Bob can determine the value of $i$ with a standard measurement.
\end{proof}

Concatenating the two proofs yields an interesting observation about one-way LOCC and maximally-entangled states:

\begin{cor}
Let ${\mathcal S} = \{ \ket{\phi_i} = (I \ot U_i)\ket{\Phi}\}$ be a set of mutually orthogonal maximally-entangled $\mathcal H_a \ot \mathcal H_b$ states with $a \le b$, and let ${\mathcal C}$ be the code space spanned by the $\{\ket{\phi_i}\}$.

If the elements of ${\mathcal S}$ can be perfectly distinguished with one-way LOCC, then {\it any} basis of ${\mathcal C}$ can also be distinguished with one-way LOCC using the same initial measurement on Alice's system; and for any $\ket{\phi}$ in ${\mathcal C}$, we can teleport the joint state $\ket{\phi}$ to Bob's system.
\end{cor}
This gives a contrast to the examples in \cite{watrous2005bipartite, duan2009distinguishability} of spaces for which no basis is locally distinguishable (which are discussed more in the next section).  In this case, we have spaces for which {\it any} basis can be distinguished with one-way LOCC; and this is a necessary condition for a set of maximally-entangled states to be distinguished with one-way LOCC.


The familiar example of quantum teleportation \cite{Teleportation} can be seen as an illustration of the results in this section.

\begin{exa}
Let ${\mathcal B} = \{ \ket{\phi_{i,j}} =  (I \ot X^iZ^j) \ket{\Phi}\}$ be the generalized Bell basis for  $\C^n \ot \C^n$. If Alice and Bob share an unknown element in $\mathcal B$ plus an additional canonical maximally-entangled state $\ket{\Phi}$, then the state of their system is in ${\mathcal S} = \ket{\Phi} \ot {\mathcal B} = \{ \ket{\Phi}_{A_1B_1} \ot \ket{\phi_{i,j}}_{A_2B_2}\}$. It is easy to show that if Alice measures her bipartite system in the   basis ${\mathcal B}$, then Bob can complete the measurement and determine the value of $(i,j)$. If we let $\mathcal C$ be the span of the elements in ${\mathcal S}$, we get \bee P_{\mathcal C} &=& \kb{\Phi}{\Phi}_{A_1B_1} \ot I_{A_2  B_2} \\ {\mathcal L}({\mathcal C}) &=& \kb{\Phi}{\Phi} \otimes \mathcal L(\H_{AB}). \eee

By Theorem \ref{thm: MaxEnt LOCC implies correctability}, Alice's measurement $\Phi_{QC}$ is correctable on the entire code space; there exists a correction operator $\mathcal R$ so that
\bee
\mathcal R \circ \left(\Phi_{QC} \otimes \mathrm{id}_B\right) (\rho) = \rho
\eee
for every $\rho \in \kb{\Phi}{\Phi} \otimes \mathcal L(\H_{AB})$. Indeed, this is the familiar protocol for quantum teleportation.

Applying  Theorem \ref{thm: correctability implies LOCC} to this protocol  affirms that, since we can teleport Alice's half of an unknown element of ${\mathcal C}$ with one-way LOCC, any orthogonal basis of ${\mathcal C}$ can be subsequently distinguished by Bob.\end{exa}

\section{More General Correction Algebras and One-Way LOCC Applications}

Expanding our investigation,  this section considers more general error correction contexts and connections with LOCC. We shall work within the general framework called ``operator algebra quantum error correction'' (OAQEC), which includes standard Knill-Laflamme quantum error correction as a special case but also includes classical and hybrid classical-quantum error correction as other distinguished special cases. To avoid having to introduce extra nomenclature, at the end of this section (see Remark~\ref{oaqecremark}) we briefly draw the connection with OAQEC more explicitly and provide some literature entrance points for the subject.

We next consider the case of correctable commutative algebras and extracting classical information in the LOCC setting.

\subsection{Extracting classical information} It is evident that quantum error correction and LOCC are not equivalent theories in general. Indeed, on the one hand the Knill-Laflamme conditions are not sufficient for LOCC purposes; Bob does not have access to the output of all of Alice's classical information, just a single outcome. But on the other hand the conditions are seen to be too strong in their full generality; in quantum state discrimination, we need only identify a set of initial states $\mathcal S$, not an entire subspace. It is this last point we focus on now, making the distinction between a correctable subspace, in which any element of a subspace can be recovered, and a correctable  set $\mathcal S$, which is correctable if any element of $\mathcal S$ can be recovered with certainty.

\begin{prop}\label{commutecorrect}
Let $\mathcal C$ be a $d$-dimensional subspace of a Hilbert space $\mathcal H$, and let $\mathcal E(\rho) = \sum_i A_i \rho A_i^*$ be a channel on $\mathcal H$.

If $\mathcal C$ has a basis $\mathcal S = \{ \ket{\phi_k}\}$ such that the set of states of $\mathcal S$ is individually correctable for $\mathcal E$, then the set of operators $\{ P_{\mathcal C} A_j^*A_i  P_{\mathcal C}\}_{i,j}$ are mutually commuting. In particular, the dimension of the span of $\{ P_{\mathcal C} A_j^*A_i  P_{\mathcal C}\}$ is at most $d$.

Conversely, if the operators $\{ P_{\mathcal C} A_j^*A_i  P_{\mathcal C}\}_{i,j}$ are normal, non-zero and mutually commuting, then there is a basis for $\mathcal C$ that is individually correctable for $\mathcal E$.
\end{prop}

\begin{proof}
The states of $\mathcal S$ being correctable for $\mathcal E$ is equivalent to the identities:
\begin{equation}\label{commute1}
\tr (\mathcal E(\kb{\phi_k}{\phi_k})\mathcal E(\kb{\phi_l}{\phi_l})) = 0,
\end{equation}
whenever $k \ne l$. Since the output of the  channel is written as a sum of positive semidefinite matrices, this is equivalent to the statement that for every $i,j$ and $k \ne l$,
\begin{equation}\label{commute2}
0 = \tr ( A_j \kb{\phi_l}{\phi_l} A_j^*A_i \kb{\phi_k}{\phi_k} A_i^* ) = \left \vert \bra{\phi_l} A_j^*A_i \ket{\phi_k} \right\vert^2 .
\end{equation}
This implies that for all $i,j$, the matrix representation for $P_{\mathcal C} A_j^*A_i P_{\mathcal C}$ is diagonal in the basis $\mathcal S$ when $\mathcal S$ is correctable for $\mathcal E$.

Conversely, if the operators $\{ P_{\mathcal C} A_j^*A_i  P_{\mathcal C}\}_{i,j}$ are normal and mutually commuting, then they are simultaneously diagonalizable by the spectral theorem, and so there is a basis $\mathcal S = \{ \ket{\phi_k} \}$ for $\mathcal C$ and scalars $\lambda_k^{(i,j)}$ such that $P_{\mathcal C} A_j^*A_i  P_{\mathcal C} = \sum_{k=1}^d \lambda_k^{(i,j)} \kb{\phi_k}{\phi_k}$. Thus, the above argument may be reversed to show that first Eqs.~(\ref{commute2}) hold and then Eqs.~(\ref{commute1}) as well.
\end{proof}



Let us consider applications of this perspective. In the  example of Watrous \cite{watrous2005bipartite}, it was shown that the orthogonal complement of a maximally-entangled state in two-qutrit space cannot be distinguished with LOCC (not just one-way LOCC). This was generalized in \cite{duan2009distinguishability}. Here, we  connect their result to our current discussion and obtain a new proof. We point to \cite{nielsen} for basics on the Schmidt decomposition and rank of a bipartite pure state.


\begin{cor}\cite{duan2009distinguishability}
Let ${\mathcal C}$ be the orthogonal complement of a state  $\ket{\phi} \in \C^d \otimes \C^d$.

No basis for $\mathcal C$ can be distinguished with one-way LOCC if the Schmidt rank of $\ket{\phi}$ is greater than two.
\end{cor}

\begin{proof}
Let $P = I - \kb{\phi}{\phi}$ be the projection onto the space $\mathcal C$. Suppose that $\mathcal C$ has a basis of states that can be distinguished with one-way LOCC. Then there exists a measurement on Alice's system, denoted by $\{ m_j \kb{\varphi_j}{\varphi_j} \}_{j = 1}^r$, such that the operators $\{Q_j:= P (\kb{\varphi_j}{\varphi_j} \otimes I_B) P  \}$ are mutually commuting and share a basis of eigenvectors.

Suppose that $\ket{\phi}$ is not a product state. Then for each $j$, $(\kb{\varphi_j}{\varphi_j} \otimes I_B) \ket{\phi}$ is not a nonzero multiple of  $\ket{\phi}$. Without loss of generality, we assume that $(\kb{\varphi_1}{\varphi_1} \otimes I_B) \ket{\phi} \ne 0$.  This implies that there must exist an eigenvector $\ket{\psi}$ of $Q_1$ such that $\bk{\phi}{\psi} = 0$ but  $\alpha_1 := \bra{\phi} (\kb{\varphi_1}{\varphi_1} \otimes I_B) \ket{\psi} \ne 0$. Writing $\ket{\psi}= (I_A \ot B)\ket{\Phi}$ and $Q_j\ket{\psi} = \mu_j\ket{\psi}$, we get
 \bee
Q_1 \ket{\psi} &= & (\kb{\varphi_1}{\varphi_1} \otimes I_B)\ket{\psi} - \kb{\phi}{\phi}(\kb{\varphi_1}{\varphi_1} \otimes I_B)\ket{\psi}\\
\mu_1 \ket{\psi} & = & \ket{\varphi_1} \ot B\ket{\overline{\varphi}_1} - \alpha_1 \ket{\phi}.  \eee
If $\ket{\psi}$ is a product state, then $\ket{\phi}$ is a linear combination of two product states and has Schmidt rank at most 2. (Here we use that $\alpha_1 \ne 0$.)  If $\ket{\psi}$ is not a product state, then for some $x \ne 1$, $B\ket{\overline\varphi_x} \ne 0$ and $B\ket{\overline\varphi_1}$ and  $B\ket{\overline\varphi_x}$ are linearly independent. By the same calculation above, $\ket{\varphi_x} \ot B\ket{\overline{\varphi}_x}$ is in the span of $\ket{\phi}$ and $\ket{\psi}$. Hence we have
\bee
\mbox{span}\{\ket{\varphi_1} \ot B\ket{\overline{\varphi}_1},\ket{\varphi_x} \ot B\ket{\overline{\varphi}_x}\} \subseteq \mbox{span} \{\ket{\phi} ,\ket{\psi}\}.
\eee
Since the vectors on the left are linearly independent, the two spans are the same and we have
\bee
\ket{\phi} \in \mbox{span} \{\ket{\varphi_1} \ot B\ket{\overline{\varphi}_1},\ket{\varphi_x} \ot B\ket{\overline{\varphi}_x}\} .
\eee
Thus $\ket{\phi}$ is a linear combination of two product states and has Schmidt rank at most 2.
\end{proof}

Proposition \ref{commutecorrect} suggests ways to find LOCC-distinguishable bases of given subsystems; or to show that they don't exist. One promising application would be the following numerical conjecture of King: 

\begin{con}\label{KingConjuecture}\cite{king2007existence} Let ${\mathcal C}$ be a three-dimensional subspace of $\C^3 \otimes \C^n$, with $n \ge 3$. Then there is
an orthonormal basis of ${\mathcal C}$ which can be reliably distinguished using one-way LOCC, where
measurements are made first on $\C^3$ and the result used to select the optimal measurement
on $\C^n$ .
\end{con}

We can restate this conjecture in the language of our results:
\begin{con} Let ${\mathcal P}$ be the projection onto  a three-dimensional subspace of $\C^3 \otimes \C^n$, with $n \ge 3$. Then there exist orthogonal states $\ket{ \varphi_1}$ and $\ket{ \varphi_2}$ in $\C^3$ such that ${\mathcal P} \left( \kb{ \varphi_1}{\varphi_1} \otimes I_n\right)  {\mathcal P}$ and ${\mathcal P} \left( \kb{ \varphi_2}{\varphi_2} \otimes I_n\right)  {\mathcal P}$ commute.
\end{con}
The equivalence of the conjectures follows immediately from the fact that if the two operators commute with each other, then they also commute with ${\mathcal P} \left( \kb{ \varphi_3}{\varphi_3} \otimes I_n\right)  {\mathcal P}$ if $\sum_i \kb{ \varphi_i}{\varphi_i}= I_3$.

We can use our results to establish sufficient conditions under which the conjectured basis exists. To do this, we look at the subspace $\mathcal X := \mbox{span} \{ B_i^*B_j \}_{i,j=1}^3$ and the associated map $\Psi: M_3 \rightarrow M_3$ defined by
\bee
\Psi(\tau) =\sum_{i,j} \kb{i}{\Phi_i}(\tau \ot I)\kb{\Phi_j}{j} = \frac{1}{3} \sum_{i,j} \kb{i}{j} \tr \tau^T B_i^*B_j .
\eee
If $\Psi$ is trace-preserving, then it is the complementary channel of the CP map $\Phi(\tau) = \frac{1}{3} \sum_i B_i\tau B_i^*$ \cite{holevo2007complementary,king2005properties}, and in this case the set $\mathcal X$ is the non-commutative graph associated with the channel $\Phi$ \cite{duan2013zero}. In general, $\Psi$ is  not a channel, but the analogy with complementary channels and non-commutative graphs is still relevant. Note also that the operator system $\mathfrak{S}_0$  generated by $\{ B_i^*B_j: i \ne j \}$ was studied in \cite{kribs2016operator} to establish necessary and sufficient conditions for one-way LOCC discrimination of a {\it particular} basis of ${\mathcal C}$; these results are expanded in Section \ref{QEC Formalism}. In the current discussion, $\mathfrak{S}_0 \subset \mathcal X$ whenever $\mathcal X$ contains the identity.

We use the subspace $\mathcal X$ to prove to following condition:
\begin{prop}\label{Subspace condition for distinguishability}
Let ${\mathcal C}$ be a $3$-dimensional subspace of $\C^3 \ot C^n$ spanned by the orthogonal basis $\{ (I \ot B_i)\ket{\Phi} \}_{i = 1}^3$, and consider the set
\bee
\mathcal X = \mbox{span}(\{ B_i^*B_j \}_{i,j=1}^3)
\eee
If $\mathcal X$ is a strict subspace of the $3 \times 3$ matrices $M_3$, then there exists a basis of  ${\mathcal C}$ that is distinguishable with one-way LOCC; and this basis may be easily computed.
\end{prop}
\begin{proof}
Suppose that $\mathcal X$ is a strict subspace of $M_3$. Then there exists a non-zero matrix $M$ with $\tr M^T B_i^*B_j =0$ for all $i,j$. We note that if such an $M$ exists, it is easily computed; and that $\Psi(M) = 0$.

We note that $\tr \Psi(M) = \frac{1}{3} \sum_{i} \tr \tau^T B_i^*B_i = 0$ implies that $M$ has both positive and negative eigenvalues. Without loss of generality, we write  $M = \sum_{i = 0}^2 \lambda_i \kb{\phi_i}{\phi_i} $  with $\lambda_0 <0 \le  \lambda_1 \le \lambda_2$.  We then use the fact that $\Psi$ is unital to write
\bee
I = \Psi(I) &=& \sum_i \Psi(  \kb{\phi_i}{\phi_i} ) \\ &=& \Psi(  \kb{\phi_1}{\phi_1} ) (1 - \frac{\lambda_1}{\lambda_0}) + \Psi(  \kb{\phi_2}{\phi_2} ) (1 - \frac{\lambda_2}{\lambda_0}).
\eee
Since $(1 - \frac{\lambda_i}{\lambda_0}) \ge 1 >0$, this implies that  $ \Psi(  \kb{\phi_2}{\phi_2} )$ and $ \Psi(  \kb{\phi_1}{\phi_1} ) $ commute, which in turn implies that if  Alice measures in the basis $\{ \ket{\phi_i} \}$, it will distinguish the basis of ${\mathcal C}$ corresponding to the common eigenbasis  of  $\{ \Psi(  \kb{\phi_i}{\phi_i} )\}$.
\end{proof}

We can see this as a way to apply  Proposition \ref{commutecorrect}  in the following example.

\begin{exa} Let ${\mathcal C}$ be the span of the following states in  $\C^3 \ot \C^3$, where $\omega$ is a primitive cube root of unity here, 
\bee
\ket{\Phi_1} &=& \frac{1}{\sqrt{3}} \left(\ket{00} + \ket{11} + \ket{22}\right) \\
\ket{\Phi_2} &=& \frac{1}{\sqrt{3}} \left(\ket{00} + \omega \ket{11} + \omega^2 \ket{22}\right) \\
\ket{\Phi_3} &=& \frac{1}{\sqrt{2}} \left( \ket{10} - \ket{01}\right).
\eee
In this case, $\dim \mathcal X = 7$, so $\Psi$ is zero on a two-dimensional space. One solution with  $\Psi(M) = 0$ is $ M = \begin{pmatrix} 0 & 0 & 1 \cr 0 & 0 & 0 \cr 1 & 0 & 0 \end{pmatrix}$, which has eigenvectors 
\[
\{ \ket{\varphi_i}\} = \{ \begin{pmatrix}1 & 0 & 1\end{pmatrix}^T ,\begin{pmatrix}1 & 0 & -1\end{pmatrix}^T ,\begin{pmatrix}0 & 1 & 0 \end{pmatrix}^T \}. 
\]
In this case, the common eigenbasis for $\Psi( \kb{\varphi_i}{\varphi_i} )$ is 
\[
\{ \begin{pmatrix}1 & \omega^2 & 0 \end{pmatrix}^T, \begin{pmatrix}1 & -\omega^2 & 0 \end{pmatrix}^T ,\begin{pmatrix}0 & 0 & 1\end{pmatrix}^T \},
\]
indicating that if Alice measures in the basis $\{ \ket{\varphi_i}\} $, we can distinguish the basis given by $\{ \ket{\Phi'_i} \}$ with
\bee
\ket{\Phi'_1} &=& \frac{1}{\sqrt{6}} \left( \ket{\Phi_1} + \omega^2 \ket{\Phi_2} \right) = \frac{1}{\sqrt{2}} \left( -\omega \ket{00} + 2\ket{11} -\omega^2\ket{22}\right) \\
\ket{\Phi'_2} &=& \frac{1}{\sqrt{6}} \left( \ket{\Phi_1} - \omega^2 \ket{\Phi_2} \right) =\sqrt{\frac{3}{2}}  \left( e^{\pi i/6}  \ket{00} + e^{-\pi i/6} \ket{22}\right) \\
\ket{\Phi'_3} &=& \ket{\Phi_3} = \frac{1}{\sqrt{2}} \left( \ket{10} - \ket{01}\right).
\eee
\end{exa}

The flexibility afforded in the example leads to the following:
\begin{cor}
Let ${\mathcal C}$ be a $3$-dimensional subspace of $\C^3 \ot C^n$ that contains two orthogonal maximally-entangled states. Then  ${\mathcal C}$ has a basis that can be distinguished with one-way LOCC.
\end{cor}
The corollary follows from the fact that if $B_1$ and $B_2$ are both maximally entangled, then $B_1^*B_1 = B_2^*B_2 = I$. This means $\dim \mathcal X < 9$, allowing us to apply Proposition \ref{Subspace condition for distinguishability}.

\subsection{The Case of Matrix Algebras} To say that  the  operators in Proposition \ref{commutecorrect} commute is to say that they are contained in a commutative algebra of simultaneously diagonalizable matrices. What happens if they lie in a more general algebra? We can generalize the above results somewhat.

Notationally below, we write $M_d$ for the algebra of $d \times d$ complex matrices.

\begin{prop}\label{prop:mixed states}
Let ${\mathcal S} = \{\rho_1, \ldots, \rho_n\}$ be orthogonal (possibly mixed) states in a finite-dimensional Hilbert space $\mathcal H$, and let $\mathcal C$ be the orthogonal direct sum of their supports, so that $\dim \mathcal C = \sum_{i=1}^n d_i$ with $d_i = \mathrm{rank}(\rho_i)$. Let  $\mathcal E$ be a channel on $\mathcal H$.

Suppose there exists a  correction operation $\mathcal R$ such that $(\mathcal R \circ {\mathcal E})(\rho_i) = \rho_i$  for all $i$. Then for any Kraus representation of $\mathcal E$ as ${\mathcal E}(\rho) = \sum_i A_i\rho A_i^*$, the set of matrices $\{ P_{\mathcal C} A_j^*A_i  P_{\mathcal C}\}_{i,j}$ are contained in an algebra $\mathcal A$ that is unitarily equivalent to an orthogonal algebra direct sum $\mathcal A \cong \oplus_i M_{d_i}$.

Conversely, if the set $\{ P_{\mathcal C} A_j^*A_i  P_{\mathcal C}\}_{i,j}$ is contained in an algebra $\mathcal A \cong \oplus_i {\mathcal M}_{d_i}$, then we can find subspaces ${\mathcal C}_i$ such that $\dim {\mathcal C}_i = d_i$; $\mathcal{C} = \oplus_i {\mathcal C}_i$; and the set of states ${\mathcal S} = \{\rho_1, \ldots, \rho_n\}$ can be corrected whenever the support of $\rho_i$ is contained in $\mathcal{C}_i$.
\end{prop}

\begin{proof}
Suppose that the map ${\mathcal E}$ is correctable on  the states $\{\rho_i\}$. Then the states $\{ {\mathcal E}(\rho_i)\}$ are mutually orthogonal. For each $i$, let ${\mathcal C}_i$ be the support of $\rho_i$, and so the algebra of operators supported on $\mathcal C_i$ is unitarily equivalent to $M_{d_i}$. Since $\tr (\rho_i\rho_j) = 0$, we have ${\mathcal C}_i \cap {\mathcal C}_j = \{0\}$ whenever $i \ne j$.  If ${\mathcal C}$ is the support of the whole set of $\rho_i$, then
\bee
{\mathcal C} = \oplus_{i = 1}^n {\mathcal C}_i .
\eee
Writing $P_{\mathcal C}$ as the projection on ${\mathcal C}$, perfect correctability means that $\tr ({\mathcal E}(\rho_i){\mathcal E}(\rho_j)) = 0$, which implies that for $\ket{\phi_i} \in {\mathcal C}_i$ and $\ket{\phi_j} \in {\mathcal C}_j$,
\bee
0 & = & \tr ({\mathcal E}(\kb{\phi_i}{\phi_i}) {\mathcal E}(\kb{\phi_j}{\phi_j}))  \\
& = & \sum_{k,l} \tr ( A_k\kb{\phi_i}{\phi_i}A_k^*A_l\kb{\phi_j}{\phi_j}A_l^* ) \\
& = & \sum_{k,l} \left \vert \bra{\phi_i}A_k^*A_l\ket{\phi_j} \right\vert^2 \eee
whenever $i \ne j$. This implies that each term in the above sum is equal to zero for any Kraus representation of $\mathcal E$.  It follows that the matrix representations for all the operators $P_{\mathcal C} A_k^*A_l  P_{\mathcal C}$ are block diagonal with respect to the decomposition ${\mathcal C} = \oplus_{i = 1}^n {\mathcal C}_i$.

The converse can be proved straightforwardly by reversing the above calculations.
\end{proof}

We may apply this result to one-way LOCC. First we establish the following identity.

\begin{lem}
Let $\H = \C^d \ot \C^d$ and let  $T = \left(\kb{\Phi}{\Phi}\right)^{PT}$ be the partial transpose of a standard maximally-entangled state on $\H$. Then
for any  state $\sigma$ on $\C^d$ we have,
\bee
d^2T\left( \sigma \otimes I \right) T =  I \otimes \sigma .
\eee
\end{lem}

\begin{proof}
By direct calculation we have:
\bee
d^2T\left( \sigma \otimes I \right) T &=& \sum_{i,j,k,l} \kb{i}{j}\sigma\kb{k}{l} \otimes \kb{j}{i} I \kb{l}{k} \\
& = & \sum_{i,j,k} \kb{i}{i} \otimes \kb{j}{j}\sigma\kb{k}{k} = I \otimes \sigma .
\eee
\end{proof}

\begin{exa}
Let $\H = \C^d \ot \C^d$. We can write the space of operators on $\H$ as a direct sum of its symmetric and antisymmetric parts, writing $P_{\mathcal C}  = I_\H = \Pi_s \oplus \Pi_a$, where $\Pi_s$ and $\Pi_a$ project onto the symmetric and antisymmetric subspaces. It is well known that the states $\rho_1 = \frac{2}{d(d+1)} \Pi_s$ and $\rho_2 = \frac{2}{d(d-1)} \Pi_a$ cannot be distinguished even with PPT measurements \cite{eggeling2002hiding,matthews2009distinguishability}.

Let $\rho_1$ and $\rho_2$ be symmetric and antisymmetric states, respectively, that are full rank. That is, $\Pi_s\rho_1=\rho_1; \Pi_a\rho_2 = \rho_2$; rank $\rho_1 = \frac{d(d+1)}{2}$; and rank $\rho_2 = \frac{d(d-1)}{2}$. Thus, $\Pi_s$ and $\Pi_a$ are projections onto the support of $\rho_1$ and $\rho_2$.  As an application of Proposition \ref{prop:mixed states}, we can show show that these states cannot be distinguished with one-way LOCC.

As above, we write $T = \left(\kb{\Phi}{\Phi}\right)^{PT}$, implying that $\Pi_s = \frac{I + dT}{2}$ and $\Pi_a = \frac{I - dT}{2}$. We represent Alice's measurement as a channel ${\mathcal E}$ with Kraus operators  $\{A_j = \sigma_j \otimes I \}$. A necessary condition for one-way distinguishability of $\rho_1$ and $\rho_2$ is that $P_{\mathcal C} A_j^*A_jP_{\mathcal C} = \Pi_s A_j^*A_j \Pi_s \oplus  \Pi_a A_j^*A_j \Pi_a$ is block diagonal, implying that the sum of the off-diagonal entries is zero.
\bee
0 & = &  \Pi_s \left(  \sigma_j  \otimes I \right) \Pi_a + \Pi_a \left(   \sigma_j  \otimes I \right) \Pi_s \\
 & = &(I + dT)\left( \sigma_j \otimes I \right)(I -dT) + (I - dT)\left( \sigma_j \otimes I \right)(I + dT) \\
 & = &   \left( \sigma_j \otimes I \right) + d^2T \left( \sigma_j  \otimes I \right)T \\
& = &   \sigma_j  \otimes I + I \otimes \sigma_j
\eee
This necessary condition is clearly not met for any  state $\sigma$; hence one-way LOCC is not possible.
\end{exa}





\begin{rmk}\label{oaqecremark}
{\rm As noted above, the error correction results in this section, though obtained with applications to LOCC as primary motivation here, may nevertheless also be viewed within the OAQEC framework as distinguished special cases.  Indeed, from the operator algebra perspective Proposition~\ref{commutecorrect}  may be viewed as the case in which the correctable algebras are commutative (finite-dimensional) von Neumann algebras. Further, Proposition~\ref{prop:mixed states} corresponds to the case in which the correctable algebras are direct sums of full (and unampliated, or untensored) matrix algebras. The papers  \cite{bkk07,bkk07b,bkk09}, including the references therein and forward references available online, give entrance points into the related error correction literature. }
\end{rmk}

\section{Distinguishable Sets of States from the Stabilizer Formalism}\label{QEC Formalism}

In \cite{kribs2016operator} we began to develop new techniques for distinguishing sets of quantum states via one-way LOCC, based on the theory of operator algebras and aspects of matrix theory, taking motivation from mathematical conditions for one-way LOCC \cite{N13}. In the spirit of the current investigation, here we expand this approach to sets of states that arise in the central quantum error correction setting of the stablizer formalism \cite{gottesmanstab,nielsen}. Meshing with our notation above, suppose $\{B_i\}$ are operators on $\mathbb{C}^d$ and $\ket{\Phi}\in\mathbb{C}^d\otimes\mathbb{C}^d$ is a fixed maximally entangled state.

We begin by making a small but significant improvement to a key result of \cite{kribs2016operator}. We recall that an {\it operator system} is a self-adjoint subspace of operators containing the identity, and a {\it separating vector} for an operator algebra $\mathfrak{A}$ is a vector $\ket{\psi}$ such that $A\ket{\psi} = 0$ with $A\in\mathfrak{A}$ only if $A=0$.  We also recall from the theory of finite-dimensional C$^*$-algebras \cite{davidson}, that every such algebra is unitarily equivalent to an algebra $\mathfrak{A}$ for which there exist positive integers $m_k$, $n_k$ such that $\mathfrak{A} =  \oplus_k (M_{m_k}\otimes I_{n_k})$.

Consider the operator system $\mathfrak{S}_0 = \mathrm{span}\,\{B_i^* B_j, I \}_{i \neq j}$. If $\mathcal S$ is distinguishable by one-way LOCC, then $\mathfrak{S}_0$ is contained in an operator system of the specific form investigated in Section~4 of \cite{kribs2016operator}, and hence by Theorem~1 in that paper, every subalgebra of $\mathfrak{S}_0$ has a separating vector.  In fact, we have the following:

\begin{thm}\label{thm2}
Suppose the operator system $\mathfrak{S}_0 = \mathrm{span}\,\{B_i^* B_j, I \}_{i \neq j}$ is closed under multiplication and hence is a C$^*$-algebra. Then $\mathcal S = \{(I\otimes B_i)\ket{\Phi}\}$ is distinguishable by one-way LOCC if and only if  $\mathfrak{S}_0$ has a separating vector.
\end{thm}

Before proving this improvement we need the following result from \cite{pereira2}.

\begin{lem} \label{thesisthm}  \cite[Theorem 3.2.4]{pereira2} Let $\mathfrak{A}=\oplus_{k} (M_{m_k}\otimes I_{n_k})$. Then the following are equivalent.

\begin{enumerate}

\item There exists a $r \times r$ unitary matrix $U$ where $r=\sum_k m_k n_k$ such that $U\mathfrak{A} U^*$  is a C$^*$-algebra of matrices with constant diagonal.

\item The algebra $\mathfrak{A}$ has a separating vector.

\item For all $k$, $n_k\geq m_k$.

\end{enumerate}

 \end{lem}


We shall describe the proofs of two implications in this result as they are short and indicate how the notions are connected. The proof of the remaining implication is more involved.

For the implication $(1)\implies (2)$:  If $U\mathfrak{A} U^*$  is a C$^*$-algebra of matrices in a basis $\{ \ket{0},\ket{1},\ldots\}$ with constant diagonal and $A\in \mathfrak{A}$ is non-zero, then $\bra{0} UA^*AU^* \ket{0}=\frac{1}{r}\tr(AA^*)>0$ and hence $AU^*\ket{0} \neq 0$ so $U^*\ket{0}$ is a separating vector for $\mathfrak{A}$.

For the implication $(2) \implies (3)$: Suppose $\mathfrak{A}$ has a separating vector.  Since $\mathfrak{A}=\oplus_{k} (M_{m_k}\otimes I_{n_k})$, each component $M_{m_k}\otimes I_{n_k}$ must have a separating vector in $\mathbb{C}^{m_kn_k}$, call this separating vector $\ket{\psi_k}$.  Then the map $M_{m_k}\otimes I_{n_k} \to \mathbb{C}^{m_kn_k}$ defined by mapping $A\in M_{m_k}\otimes I_{n_k}$ to $A\ket{\psi_k}$ must be an injective linear map. Therefore, $m_k^2=\dim(M_{m_k}\otimes I_{n_k})\le \dim(\mathbb{C}^{m_kn_k})=m_kn_k$ and hence $m_k\le n_k$ for all $k$.

The full proof that $(3) \implies (1)$  is rather intricate and can be found in its entirety in \cite{pereira2} (along with a stronger version of Lemma \ref{thesisthm}). We omit the proof here, aside from pointing out that in the simple case  $\mathfrak{A} = M_m \ot I_m$, the algebra  $U\mathfrak{A} U^*$  has constant diagonal if the rows of $U$ are the  $m$-dimensional Pauli states.
\strut

We now prove Theorem \ref{thm2}.

\begin{proof}
Suppose $\mathfrak{S}_0$ has a separating vector. Then by Lemma~\ref{thesisthm}, there is a unitary $U$ such that $U\mathfrak{S}_0 U^*$ is a C$^*$-algebra of matrices with constant diagonal. Hence we have the condition from \cite{N13} (also derived in \cite{kribs2016operator}) satisfied, and $\mathcal S$ is one-way distinguishable. The converse follows from the discussion preceding the theorem statement.
\end{proof}

\begin{rmk}
{\rm We know there exist positive integers $m_k$, $n_k$ such that the unitary equivalence $\mathfrak{A} = \mathrm{Alg}\,(\mathfrak{S}_0) \cong \oplus_k (M_{m_k}\otimes I_{n_k})$ holds, and that by Theorem~\ref{thesisthm}, $\mathfrak{A}$ has a separating vector if and only if $n_k\geq m_k$ for all $k$. Thus, as a road map to examples of sets of indistinguishable states, we can look for sets $\{B_i\}$ such that $\mathfrak{S}_0 = \mathfrak{A}$ and $m_k > n_k$ for some $k$.

So we are led to consider sets of unitaries $\{B_i\}$, such that the set is closed under multiplication, taking adjoints, and taking inverses (up to scalar multiples). We also recall that the states $\mathcal S = \{ (I\otimes B_i)\ket{\Phi}\}$ are all maximally entangled precisely when each $B_i$ is unitary. }
\end{rmk}

Before continuing, we show how Theorem \ref{thm2} gives an alternate proof of a result from \cite{kribs2016operator}.  We first need the following definition.

\begin{defn}
We say that a set of states $\{(I\otimes B_k)\ket{\Phi}\}_{k=1}^n$ have a simultaneous Schmidt decomposition if there exists two unitary matrices $U$ and $V$ and $n$ complex diagonal matrices $D_k$ such that  for each $k$, $B_k=UD_kV$.

\end{defn}

\begin{cor} \cite[Proposition 3]{kribs2016operator} Any set of orthonormal states $\{(I\otimes B_k)\ket{\Phi}\}_{k=1}^n$ which have a simultaneous Schmidt decomposition are distinguishable by one-way LOCC. \end{cor}

\begin{proof} Without loss of generality, we may assume that $\{(I\otimes B_k)\ket{\Phi}\}_{k=1}^n$ is a maximal set of orthonormal states which have a simultaneous Schmidt decomposition.  Therefore there exists unitaries $U$ and $V$ and $n$ complex diagonal matrices $D_k$ such that for each $k$, $B_k=UD_kV$.  Then the operator system $\mathfrak{S}_0 = \mathrm{span}\,\{B_i^* B_j, I \}_{i \neq j}=V^*\Delta V$ where $\Delta$ is the algebra of diagonal matrices.  Since $V^*\Delta V$ has a separating vector, the result follows from Theorem \ref{thm2}.  \end{proof}

\subsection{States from the Stabilizer Formalism}

We will use the following notation below: if $S$ is a subset of a group $G$, then $\langle S\rangle$ is the subgroup of $G$ generated by $S$.
Let $\mathcal{P}_{n}$ be the $n$-qubit Pauli group; that is, the group with generating set as follows: $$\mathcal{P}_{n} : = \langle  \pm iI; X_{j}, Y_j, Z_{j} : 1 \leq j \leq n \rangle\,,$$
where  $X_1 = X\otimes I\otimes \cdots \otimes I = X \otimes I^{\otimes (n-1)}$, etc.

Next we state a basic result on maximal abelian subgroups of the Pauli group (e.g. see \cite{mousavi2016quantum}). We recall the {\it Clifford group} is the normalizer subgroup of the Pauli group $\mathcal P_n$ inside the group of $n$-qubit unitary operators.

\begin{lem}\label{lem12} \cite[Lemma 2.4]{mousavi2016quantum}
  Let $G$ be a subgroup of $\mathcal{P}_n$ and let $S_0$ be a minimal generating set for a maximal abelian subgroup $S$ of $G$, where $S_0=\{ g_1,..., g_m \} $. Then, we can find a unitary $U$ in the Clifford group such that $U^* g_j U = Z_j$ , $1 \le j \le m$.
 \end{lem}

We can now prove the main result of this section.

\begin{thm}\label{stabform}
Let $\{ B_i \}$ be a complete set of $4^k$  encoded logical Pauli operators for a stabilizer $k$-qubit code on $n$-qubit Hilbert space. Then the set of states $\mathcal S = \{ (I\otimes B_i)\ket{\Phi} \}$ is distinguishable by one-way LOCC if and only if $k \leq \frac{n}{2}$.
\end{thm}

\begin{proof} By Lemma \ref{lem12}, it suffices to prove the result for the canonical case: we have the $4^k$ element set of $n$-qubit Pauli operators $\mathcal{P}_{n,k} =  \langle X_{j}, Z_{j} : 1 \leq j \leq k \rangle / \{ \pm iI\}$, which form a complete set of encoded operations for the code $\mathcal C = \mathrm{span}\,\{\ket{i_1 \cdots i_k 0\cdots 0} : i_j = 0,1\}\subseteq \mathbb C^{2^n}$, the stabilizer subspace for the subgroup generated by $\{ Z_{k+1}, \cdots , Z_{n}\}$.

Observe here that $\mathfrak{S}_0 := \mathrm{span}\,(\mathcal P_{n,k}) = \mathrm{Alg}\,(\mathcal P_{n,k}) = M_{2^k}\otimes I_{2^{n-k}}.$ Hence from Theorem~\ref{thm2}, the states $\mathcal S$ are distinguishable by one-way LOCC if and only if $\mathfrak{S}_0$ has a separating vector if and only if $2^k \leq 2^{n-k}$, in other words $2k \leq n$, as required.
\end{proof}

\begin{rmk}
{\rm Since the states in Theorem~\ref{stabform} are maximally entangled, we see that Theorem~\ref{stabform} saturates the bound that you cannot distinguish more than $d$ maximally-entangled states in $\C^d \otimes \C^d$ with one-way LOCC \cite{ghosh2004distinguishability, Nathanson-2005}. In this example, the dimension is $d = 2^n \ge 4^k$ when $k \le \frac{n}{2}$. The converse that one-way LOCC is possible whenever $k \le \frac{n}{2}$ is significant. As pointed out in \cite{cosentino2013small}, there exist subsets of $\mathcal{P}_{n}$ of size less than $2^n$ that are not distinguishable even with positive partial transpose operations (and hence not with one-way LOCC).}
\end{rmk}

\section{Conclusion \& Outlook}

There are a number of lines of investigation we see potentially coming out of this work. It would be interesting to further explore the channel perspective of LOCC in general and its quantum error correction connections specifically. We obtained some LOCC applications of this combined perspective here, and we feel there should be others obtainable.

Additionally, as noted in Remark~\ref{oaqecremark}, some of the LOCC applications presented here may be viewed as generated through connections with the framework for operator algebra quantum error correction, which also includes the approach for hybrid classical and quantum error correction. One significant part of the theory has not been applied here though, namely the case of correctable subsystems and so-called ``operator quantum error correction''  \cite{klp05,klpl06}. This aspect of the theory is well-developed and has strong connections with studies in passive quantum error correction, including decoherence-free and noiseless subsystems. It is not clear how correctable subsystems relate to LOCC analyses, but the results obtained here suggest this is a topic worth pursuing.

It should also be possible to build on the approach of the previous section, to construct other new sets of states that are distinguishable or indistinguishable under one-way LOCC based on the results of \cite{kribs2016operator}. For instance, one could consider other states that arise in quantum error correction, such as states generated by error operators defined through generalized representations of the Pauli relations.

\strut

{\noindent}{\it Acknowledgements.} D.W.K. was partly supported by NSERC. C.M. was partly supported by the African Institute for Mathematical Sciences and Mitacs. R.P. was partly supported by NSERC. M.N. acknowledges the ongoing support of the Saint Mary's College Office of Faculty Research, and would like to thank Andrew Conner for helpful conversations.

\bibliographystyle{plain}

\bibliography{MNBibfile}

\end{document}